\documentclass[12pt]{article}

\usepackage[T1]{fontenc}
\usepackage[utf8]{inputenc}
\usepackage{lmodern}
\usepackage{microtype}
\usepackage[margin=1.12in]{geometry}
\usepackage{amsmath,amssymb,amsthm,mathtools}
\usepackage{booktabs,array,enumitem}
\usepackage[authoryear,round]{natbib}
\usepackage{xcolor}
\usepackage{hyperref}
\hypersetup{
  colorlinks=true,
  linkcolor=blue!50!black,
  citecolor=blue!50!black,
  urlcolor=blue!50!black,
  pdftitle={Taxing Capital to Protect It},
  pdfsubject={Safe and vulnerable tax bases under limited commitment},
  pdfauthor={Georgy Lukyanov and Hengrina Ly}
}

\numberwithin{equation}{section}
\setlist{itemsep=0.25em,topsep=0.4em}

\newtheorem{assumption}{Assumption}[section]
\newtheorem{proposition}[assumption]{Proposition}
\newtheorem{theorem}[assumption]{Theorem}
\newtheorem{lemma}[assumption]{Lemma}
\newtheorem{corollary}[assumption]{Corollary}
\theoremstyle{definition}
\newtheorem{definition}[assumption]{Definition}
\newtheorem{example}[assumption]{Example}
\theoremstyle{remark}
\newtheorem{remark}[assumption]{Remark}

\newcommand{\R}{\mathbb R}

\newcommand{\argmax}{\mathop{\rm arg\,max}}
\newcommand{\wh}[1]{\widehat{#1}}
\newcommand{\dd}{\,d}
\newcommand{\cV}{\mathcal V}
\newcommand{\cJ}{\mathcal J}
\newcommand{\cR}{\mathcal R}

\title{Taxing Capital to Protect It}
\author{
Georgy Lukyanov\thanks{Toulouse School of Economics;
\href{mailto:georgy.lukyanov@tse-fr.eu}{georgy.lukyanov@tse-fr.eu}.}
\and
Hengrina Ly\thanks{Université de Rouen Normandie, LERN; ENS Paris-Saclay,
CEPS; \href{mailto:hengrina.ly@univ-rouen.fr}{hengrina.ly@univ-rouen.fr}.}
}
\date{September 2026}

\begin{document}
\maketitle

\begin{abstract}
This paper studies the composition of taxation when the government cannot fully commit to respecting private returns after investment. A fiscal authority must finance a given expenditure from labor and capital income. Ordinary tax receipts are protected, but an opportunistic executive can seize part of the remaining capital payment. We show that a revenue-neutral increase in the capital tax raises the probability of compliance whenever the labor tax is below its local, fixed-wage revenue peak, provided the equilibrium remains on a regular mixing branch. This result does not depend on the elasticity of substitution between capital and labor. The same reform can raise investment: the gain in expected retention must outweigh the decline in the opportunist's continuation gain as compliance becomes less informative. We also characterize the Ramsey allocation conditional on full compliance. Limited commitment then imposes a ceiling on sustainable capital payments, rather than a general lower bound on the statutory capital-tax rate. Finally, we distinguish a change in the authority's concern for the future from an increase in both actors' patience. The former favors more informative policies when continuation welfare is convex; the latter has no unconditional direction in the reduced-form policy problem.
\end{abstract}


\clearpage
\section{Introduction}
\label{sec:introduction}

An investor considering a project cares about both the tax that has been announced and the possibility that the government will subsequently demand more. A low statutory tax is of little comfort if, once the investment is in place, the executive can appropriate the remaining return. Conversely, a government that collects more through ordinary taxation leaves less to be taken at that later date. This raises a question: can an increase in the capital tax make investment more attractive by reducing the government's incentive to expropriate it?

The usual objection is immediate. A capital tax reduces the return received by the investor, and hence discourages the investment from which the government hopes to collect revenue. In the absence of commitment, however, this is only one part of the comparison. The return promised to investors is also a potential source of income for an opportunistic executive. Reducing that return can make it less attractive to violate the tax code, thereby increasing the probability that investors receive what they have been promised. Whether investment rises or falls depends on which of these two effects is stronger.

This paper develops a model in which the fiscal authority chooses how to finance a fixed expenditure requirement from two tax bases. The first is labor income, which is paid and taxed before the executive can intervene. The second is the return to project capital, part of which remains exposed to seizure after ordinary taxes have been collected. We refer to these as the \emph{safe} and the \emph{vulnerable} bases, respectively.\footnote{These labels concern the timing of payments and the extent of enforcement. They are not claims that labor is intrinsically safe or that all capital income is exposed. A project whose receipts remain within the executive's reach is a useful interpretation; income that has already been paid and placed beyond that reach supplies the comparison.} The authority announces the tax code before households choose investment and labor. It has no private information about the executive's type. Ordinary revenue is committed to public spending, whereas the executive retains discretion over the remaining exposed payment.

We introduce reputation by assuming that the executive may be of two types. The committed type always respects the tax code. The opportunistic type may instead seize the exposed return whenever it is profitable to do so. Citizens do not observe the executive's type, but they do observe whether it has respected the code. Seizure therefore reveals the opportunistic type, while compliance preserves the possibility that the executive is committed. The opportunist may be willing to forgo current extraction in order to retain this reputation and the future opportunities associated with it.

The separation between the authority that chooses ordinary taxes and the executive that can subsequently violate the code is important. It allows us to ask a Ramsey question about tax composition without having the choice of taxes itself signal the executive's private type. It also specifies what ordinary taxation can accomplish: a payment assigned to the public budget is no longer available for private seizure.\footnote{One possible interpretation is a budget authority that can commit ordinary receipts to expenditure, but cannot fully control an executive's extraordinary intervention in private projects. The model assumes this division of powers; it does not explain how such an institution emerges or survives.} If the executive could appropriate ordinary receipts as well, collecting them earlier would not provide the same protection.

Our first result compares two ways of raising the same revenue. Starting from an equilibrium in which the opportunist randomizes between compliance and seizure, consider a small increase in the capital-tax rate and a compensating reduction in the labor-tax rate. Theorem \ref{thm:partial-vulnerability} gives an exact condition under which this reform raises compliance. On a regular equilibrium branch, the condition is that the labor tax be below its local revenue peak when the wage is held fixed. The elasticity of substitution between capital and labor and the elasticity of capital supply affect the size of the response, but not its sign.

The intuition involves both the direct tax effect and the adjustment of production. Taxing capital collects part of the payment that the opportunist could otherwise seize. Taxing labor can also reduce that payment, but indirectly: it discourages labor supply and changes output and the return to capital. Thus it would not be enough to compare the two taxes at fixed factor supplies. Our result takes these general-equilibrium responses into account. Below the stated labor-tax threshold, the direct effect of collecting the vulnerable return is sufficiently strong that a revenue-neutral shift toward capital taxation makes compliance more likely.

A higher probability of compliance does not, by itself, establish that investors are better protected. The reform also leaves them with a smaller statutory share of the capital return. Proposition \ref{prop:partial-investment} separates these effects. Investment rises if the elasticity of expected retention with respect to compliance exceeds the elasticity of the opportunist's continuation gain with respect to reputation. The second elasticity enters because more frequent compliance is less informative: the public attributes a smaller reputational gain to an action that opportunists themselves undertake more often. When the investment condition holds, the reduction in labor taxation also raises labor supply. Output and current resource surplus then increase even though the government collects the same amount of revenue.

The distinction between a local reform and a comparison across equilibrium regimes leads to our second result. We ask which allocations can be sustained if the authority requires the executive to comply with certainty. Household optimality allows us to express the amount exposed to seizure directly in terms of investment. The incentive constraint then places an upper bound on sustainable capital, which increases with the executive's initial reputation and its patience, holding the continuation schedule fixed. On an allocation branch where surplus is single-peaked, the constrained Ramsey choice is the ordinary Ramsey capital stock or this upper bound, whichever is smaller (Proposition \ref{prop:allocation-ceiling}). The result is conditional on full compliance; it does not say that inducing full compliance is always the authority's preferred policy.

At first sight, this ceiling on investment might seem to imply a lower bound on the capital-tax rate. That conclusion is not generally valid. The tax rate needed to implement a particular allocation also depends on the marginal product of capital, which changes with both inputs. We therefore state the restriction on the \emph{capital payment and allocation}, where it has a definite direction, rather than imposing a direction on the implementing tax rate that need not follow.

Finally, we study the value of information generated by the executive's behavior. More frequent compliance improves current credibility, but makes the public less able to distinguish the two types. Whether this is a cost depends on the shape of continuation welfare. When that welfare is convex in reputation, an authority that places more weight on the future favors policies that leave more scope for learning. When it is concave, smoother reputation can instead be desirable. Our two-period production benchmark illustrates the latter case and derives the opportunist's continuation elasticity from factor shares and supply elasticities.

It matters here what we mean by greater concern for the future. Raising only the authority's continuation weight leaves the opportunist's incentives unchanged. Raising a common discount factor changes those incentives as well. Proposition \ref{prop:common-beta} identifies the additional terms in that comparison, and Example \ref{ex:common-beta-counterexample} shows why the direction of the fixed-weight result cannot be carried over without further assumptions. The example is a reduced-form policy problem, not a counterexample constructed from the production economy.

\subsection{Related literature}

Our starting point is the Ramsey problem of choosing distortionary taxes to finance an expenditure requirement \citep{Ramsey1927,DiamondMirrlees1971a,DiamondMirrlees1971b}. The departure is that the private return left after taxation also affects the executive's incentive to respect the allocation. The authority must therefore account for the effect of its tax instruments on compliance, in addition to their familiar effects on factor supply.

The closest theoretical connection is to \citet{Phelan2006} and \citet{Lu2013}. In these models, uncertainty about government type can sustain restraint by an opportunistic government. Lu allows the trustworthy type to choose its announced policy and studies the resulting interaction between taxation, credibility, and learning. Our paper does not claim this interaction as a new mechanism. We ask instead how to divide a given revenue requirement between two bases with different exposure to subsequent seizure. The main contribution is the signed comparison between the two instruments, together with its implications for investment and for the set of fully credible allocations.

A related argument appears in \citet{AcemogluGolosovTsyvinski2011}, where capital taxation can relax the politician's incentive constraint by reducing the capital stock available for extraction. The distinction is that in our model ordinary taxation removes part of the residual payment from the executive's reach. Investment need not fall for the incentive constraint to improve; under the conditions given above, it rises. Other analyses of taxation without full commitment include \citet{BenhabibRustichini1997}, \citet{PhelanStacchetti2001}, \citet{Reis2013}, \citet{Park2014}, and \citet{ScheuerWolitzky2016}. The interaction between investment and the threat of expropriation also connects our question to \citet{AguiarAmador2011}.

The institutional distinction between a rule and its subsequent enforcement relates to \citet{DovisKirpalani2021} and \citet{HalacYared2022}. We take the protected status of ordinary revenue as given and study the composition of the taxes paid into that budget. \citet{Yun2026} examines government reputation, foreign investment, corporate taxation, and profit shifting. The present paper does not include profit shifting: its comparison is between statutory taxes on safe and vulnerable factor payments at a fixed fiscal requirement.\footnote{The separate project by \citet{AblyatifovLukyanov2026} concerns government reputation and fiscal capacity. Our object here is tax composition at a given expenditure requirement, not the scale of the fiscal mandate.}

The rest of the paper is organized as follows. Section \ref{sec:environment} describes the model and the equilibrium conditions. Section \ref{sec:responses} studies the local mixing equilibrium. Section \ref{sec:composition} establishes the main tax-composition result and its implications for investment and surplus. Section \ref{sec:allocation} characterizes the constrained Ramsey allocation under full compliance. Section \ref{sec:terminal} derives a two-period continuation economy, and Section \ref{sec:dynamics} considers learning and the two patience experiments. Section \ref{sec:scope} discusses the interpretation and the limits of the results. The proofs are collected in the appendix.

\section{The model}
\label{sec:environment}

We consider a fiscal authority, an executive, and a unit mass of atomistic households. The authority must finance an exogenous expenditure \(G>0\). It can tax labor and capital income, but cannot fully prevent the executive from subsequently seizing private capital returns. The executive's type is not observed by either the authority or households. Their common prior that it is committed is \(p\in(0,1)\).

The analysis begins at a given state \((p,G)\). We first take continuation payoffs as given functions of the posterior belief about the executive. This allows us to study the current tax-composition problem without imposing a particular future economy. Section \ref{sec:terminal} subsequently derives these payoffs in a two-period example.

\subsection{Production and household choices}

Output is produced competitively using capital \(K\) and labor \(L\). Households receive the corresponding factor payments and bear the costs of supplying the two inputs. We use quasilinear utility so that a household's supply decisions depend on its marginal expected returns, without an additional wealth effect.

\begin{assumption}
\label{ass:primitives}
Production is given by \(Y=F(K,L)\). On \(\R_{++}^2\), \(F\) is twice continuously differentiable, strictly increasing, concave, strictly quasiconcave, and homogeneous of degree one. At each allocation considered, both factor shares are positive and the local elasticity of substitution \(\sigma\) is finite and strictly positive. Competitive factor prices are \(r=F_K(K,L)\) and \(w=F_L(K,L)\).

Factor-supply costs are additively separable, \(\Psi_K(K)+\Psi_L(L)\), with each function twice continuously differentiable and satisfying \(\Psi_i'>0\) and \(\Psi_i''>0\) at the interior allocations under consideration. Households take prices, aggregate statutory revenue, and any aggregate rebate as given. Define the point supply elasticities by
\begin{equation}
 \varepsilon_i
 =\frac{\Psi_i'(i)}{i\Psi_i''(i)}>0,
 \qquad i\in\{K,L\}.
 \label{eq:point-elasticities}
\end{equation}
\end{assumption}

The notation in \eqref{eq:point-elasticities} permits elasticities to vary with the allocation. We do not require isoelastic factor supply for the main result. Whenever we use \(\varepsilon_K\) or \(\varepsilon_L\) in a local comparison, it is evaluated at the allocation from which the comparison starts.

\subsection{Taxes, exposure, and timing}

The authority chooses statutory rates \((\tau_K,\tau_L)\in[0,1)^2\). For convenience, write the shares left after ordinary taxation as
\begin{equation}
 x=1-\tau_K,\qquad y=1-\tau_L.
 \label{eq:net-rates}
\end{equation}
Revenue must exactly cover the expenditure requirement:
\begin{equation}
 R\equiv\tau_KrK+\tau_LwL=G.
 \label{eq:budget}
\end{equation}
These receipts are assigned to a protected public account. They are not part of the executive's private seizure opportunity.

The sequence of decisions is as follows. First, the authority announces the tax code, knowing \((p,G)\) but not the executive's type. Households then choose capital and labor in anticipation of the executive's subsequent behavior. Production takes place, ordinary taxes are collected, and labor income is paid. Finally, the executive decides whether to respect the remaining capital payment or seize its exposed part. The public observes this decision and updates its belief.

Capital is fully depreciating project investment.\footnote{Thus the private payment at risk is the current return to the project. If the executive could also seize an undepreciated asset, its value would have to be included in the deviation payoff. The formulas below should not be read as incorporating that additional source of appropriation.} A fraction \(\xi\in(0,1]\) of its after-tax payment remains vulnerable at the executive's decision date. The additional private return from seizure is therefore
\begin{equation}
 X=\xi xrK.
 \label{eq:seizure-prize}
\end{equation}
Labor income has already been paid and is not available for seizure. The case \(\xi=1\) corresponds to full exposure of the residual capital payment; for \(\xi<1\), investors retain part of that payment even if the executive intervenes.

The committed executive, denoted by \(H\), honors the code. The opportunistic executive, denoted by \(O\), chooses whether to comply. Let \(s\) be the probability with which \(O\) complies. The overall probability of compliance is
\begin{equation}
 q=p+(1-p)s\in[p,1].
 \label{eq:q}
\end{equation}
Notice that \(q=p\) does not mean that every executive seizes the return: the committed type still complies. It means that the opportunistic type seizes with certainty.

A unit of capital earns the net return \(xr\) after compliance and \((1-\xi)xr\) after seizure. Its expected payment is consequently \(m(q)xr\), where
\begin{equation}
 m(q)=1-\xi(1-q),
 \qquad
 \zeta(q)=\frac{\dd\log m(q)}{\dd\log q}
 =\frac{\xi q}{1-\xi(1-q)}.
 \label{eq:retention}
\end{equation}
We will use \(\zeta\) to measure how responsive the expected-retention multiplier is to compliance. Household optimality gives
\begin{equation}
 \Psi_K'(K)=m(q)xr,
 \qquad
 \Psi_L'(L)=yw.
 \label{eq:household-focs}
\end{equation}
Thus labor supply is affected by the ordinary labor tax, while capital supply also depends on the likelihood of subsequent seizure. The effect on labor of a change in \(q\) operates indirectly through capital and the wage.

\subsection{Reputation and the executive's decision}

Seizure can be chosen only by the opportunist, so it reveals type \(O\). Compliance has a different informational content. The committed type complies with probability one, but the opportunist may also comply. Bayes' rule therefore gives posterior \(z=p/q\) after compliance and zero after seizure.

Let \(V_O(z)\) denote the opportunist's continuation payoff when the public's posterior is \(z\). Relative to the payoff after revelation, the gain from maintaining reputation is
\begin{equation}
 \Delta(z)=V_O(z)-V_O(0)>0,
 \qquad
 \eta(z)=\frac{z\Delta'(z)}{\Delta(z)}.
 \label{eq:delta-eta}
\end{equation}
At the positive posterior states used below, we assume that these payoff functions are continuously differentiable and that \(\Delta'(z)>0\). The elasticity \(\eta\) will be important because a change in current compliance also changes the posterior that compliance earns.

Suppose the opportunist discounts its continuation payoff by \(\beta_O\in(0,1)\). Once investment is installed, it compares the current gain \(X\) from seizure with the discounted reputational loss \(\beta_O\Delta(p/q)\). It can randomize between the two actions only if
\begin{equation}
 \xi xrK=\beta_O\Delta(p/q).
 \label{eq:mixing-ic}
\end{equation}
If \(q=p\), certain seizure by \(O\) is a best response when the left-hand side is weakly larger than the right-hand side. If \(q=1\), compliance is a best response when the inequality is reversed. Each comparison uses the allocation and posterior associated with that value of \(q\).

\begin{definition}
\label{def:equilibrium}
At a state \((p,G)\), a tax-code equilibrium consists of a statutory tax vector, an allocation, and a compliance probability satisfying exact balance \eqref{eq:budget}, household optimality \eqref{eq:household-focs}, competitive factor pricing, Bayesian updating, and the opportunist's best response in \eqref{eq:mixing-ic} or its corresponding boundary inequality. For the local results, we select a single-valued differentiable interior branch from this correspondence.
\end{definition}

This definition describes the outcomes that a tax code can support. The authority's optimal choice among them is a separate question, for which we now specify its objective. In particular, a local reform along a selected mixing branch need not be a comparison between globally optimal policies. There may also be other equilibrium branches or other compliance regimes.

\subsection{The fiscal authority}

Define current resource surplus by
\begin{equation}
 S=F(K,L)-\Psi_K(K)-\Psi_L(L).
 \label{eq:surplus}
\end{equation}
Since expenditure is fixed, its resource cost and any benefit depending only on \(G\) can be omitted from comparisons between tax codes. Let \(D=(1-q)X\) be expected diversion and let \(\lambda_D\geq0\) be the social loss per unit diverted.\footnote{With \(\lambda_D=0\), resources received by the opportunist count as a transfer within the welfare criterion. A positive \(\lambda_D\) can reflect dissipation or a lower social weight on the executive's private receipts. The investment result does not rely on choosing a positive diversion loss.}

The authority must also account for the information carried into the continuation. For a social continuation value \(\cV\), its expected value is
\begin{equation}
 \cJ(q;p)=q\cV(p/q)+(1-q)\cV(0).
 \label{eq:expected-continuation}
\end{equation}
If \(\omega\) is the weight on continuation, the authority evaluates a tax-code equilibrium by
\begin{equation}
 \mathcal W=S-\lambda_DD+\omega\cJ(q;p).
 \label{eq:authority-objective}
\end{equation}
It chooses among the fiscally feasible equilibria under the maintained selection. We initially allow \(\omega\) to vary separately from \(\beta_O\), so that a change in the authority's objective can be distinguished from a change in the executive's incentives. Section \ref{sec:dynamics} considers a common discount factor.

\section{Regular mixing equilibria}
\label{sec:responses}

Before comparing the two taxes, we need to understand how the private economy adjusts to a given code. An increase in a tax rate changes the return to supplying an input. That response affects factor prices and the executive's gain from seizure. In a mixing equilibrium, the probability of compliance must then adjust until the executive is indifferent again. This section derives these responses and states the conditions under which they describe a locally unique equilibrium.

\subsection{Factor supply at a given probability of compliance}

For any positive variable \(a\), let \(\wh a=\dd\log a\) denote its proportional change. Write
\begin{equation}
 \theta=\frac{rK}{Y}\in(0,1)
 \label{eq:capital-share}
\end{equation}
for capital's share of output. Constant returns imply that factor prices depend on the capital--labor ratio. Using the local elasticity of substitution, their proportional changes are
\begin{equation}
 \wh r=-\frac{1-\theta}{\sigma}(\wh K-\wh L),
 \qquad
 \wh w=\frac{\theta}{\sigma}(\wh K-\wh L).
 \label{eq:factor-price-responses}
\end{equation}
It is convenient to combine ordinary taxation and expected seizure into the two return multipliers
\begin{equation}
 u=m(q)x,\qquad v=y,
 \label{eq:return-shifters}
\end{equation}
and define
\begin{equation}
 \mathcal D
 =1+\frac{\varepsilon_K(1-\theta)+\varepsilon_L\theta}{\sigma}.
 \label{eq:D}
\end{equation}

\begin{lemma}
\label{lem:factor-responses}
At an interior competitive allocation satisfying Assumption \ref{ass:primitives}, the responses of factor supplies to \(u\) and \(v\) satisfy
\begin{equation}
 \begin{pmatrix}\wh K\\ \wh L\end{pmatrix}
 =
 \begin{pmatrix}c&d\\ e&f\end{pmatrix}
 \begin{pmatrix}\wh u\\ \wh v\end{pmatrix},
 \label{eq:factor-matrix}
\end{equation}
where
\begin{align}
 c&=\frac{\varepsilon_K(1+\varepsilon_L\theta/\sigma)}{\mathcal D},
 &
 d&=\frac{\varepsilon_K\varepsilon_L(1-\theta)}
 {\sigma\mathcal D},
 \notag\\
 e&=\frac{\varepsilon_K\varepsilon_L\theta}
 {\sigma\mathcal D},
 &
 f&=\frac{\varepsilon_L[1+\varepsilon_K(1-\theta)/\sigma]}
 {\mathcal D}.
 \label{eq:factor-coefficients}
\end{align}
The determinant satisfies \(cf-de=\varepsilon_K\varepsilon_L/\mathcal D\). Writing gross factor payments as \(I_K=rK\) and \(I_L=wL\), we also have
\begin{equation}
 \wh I_K=A\wh u+B\wh v,
 \qquad
 \wh I_L=C\wh u+E\wh v,
 \label{eq:income-responses}
\end{equation}
where
\begin{equation}
 A=\left(1+\frac1{\varepsilon_K}\right)c-1,\quad
 B=\left(1+\frac1{\varepsilon_K}\right)d,\quad
 C=\left(1+\frac1{\varepsilon_L}\right)e,\quad
 E=\left(1+\frac1{\varepsilon_L}\right)f-1.
 \label{eq:ABCE}
\end{equation}
These coefficients satisfy \(A+1>0\) and \(B>0\). The coefficient \(A\) can have either sign.
\end{lemma}

The positive entries in the factor-response matrix have a familiar interpretation. A higher expected return to capital encourages capital supply, and its effect on the wage also encourages labor supply. A higher net wage encourages labor supply and, through the return to capital, also encourages investment. Gross capital income need not increase with every increase in its return multiplier, because the marginal product changes as well; this explains why the sign of \(A\) is left unrestricted.

\subsection{The executive's response and regularity}

We now allow compliance to adjust. Since \(\wh u=\wh x+\zeta\wh q\), the proportional change in the seizure prize is
\[
 \wh X=A\zeta\wh q+(A+1)\wh x+B\wh y.
\]
At the same time, \eqref{eq:mixing-ic} requires \(\wh X=-\eta\wh q\) when \(p\), \(\xi\), \(\beta_O\), and \(\Delta\) are fixed. Equating the two expressions gives
\begin{equation}
 (A\zeta+\eta)\wh q=-(A+1)\wh x-B\wh y.
 \label{eq:compliance-response}
\end{equation}

To see what \eqref{eq:compliance-response} means, first hold the labor tax fixed. An increase in \(\tau_K\) lowers \(x\), reducing the payment available for seizure. On a branch where \(A\zeta+\eta>0\), compliance rises. The posterior after compliance then falls, reducing the reputational reward until the executive is indifferent again. A revenue-neutral reform also changes the labor tax, so this fixed-code derivative is only the first step in our comparison.

\begin{proposition}
\label{prop:mixing-existence}
Consider a fixed tax code with positive output. Suppose that for every \(q\in[p,1]\), household optimality has a unique solution \((K(q),L(q))\), continuous in \(q\) and differentiable on the interior. Define the difference between the logarithms of the seizure prize and the discounted continuation gain by
\begin{equation}
 \Phi(q)
 =
 \log[\xi xr(q)K(q)]
 -\log[\beta_O\Delta(p/q)].
 \label{eq:log-compliance-gap}
\end{equation}
If
\begin{equation}
 \Phi(p)<0<\Phi(1),
 \qquad
 A(q)\zeta(q)+\eta(p/q)>0
 \quad\text{for all }q\in(p,1),
 \label{eq:mixing-existence-conditions}
\end{equation}
then the mixing incentive constraint has a unique solution \(q^*\in(p,1)\). This solution is continuously differentiable in the code and prior locally wherever the primitives have the required continuous derivatives. If exact balance also holds and both total revenue derivatives are positive, the equilibrium lies on a regular local \(C^1\) balanced-budget mixing contour.
\end{proposition}

\begin{definition}
\label{def:regular}
Let
\begin{equation}
 \cR_K=\left.\frac{\dd R}{\dd\tau_K}\right|_{\tau_L,p},
 \qquad
 \cR_L=\left.\frac{\dd R}{\dd\tau_L}\right|_{\tau_K,p}
 \label{eq:revenue-derivatives}
\end{equation}
denote the total revenue derivatives, allowing factor supplies, prices, and compliance to respond. An interior mixing equilibrium is called \emph{regular} when
\begin{equation}
 A\zeta+\eta>0,\qquad \cR_K>0,\qquad \cR_L>0.
 \label{eq:regularity}
\end{equation}
\end{definition}

The first inequality ensures that the feedback between investment and compliance has the orientation used above. The two revenue inequalities mean that either tax, considered separately, still raises revenue after all equilibrium responses are included. Consequently, keeping revenue fixed requires one tax to fall when the other rises.

Proposition \ref{prop:mixing-existence} provides sufficient conditions for such a mixing response to exist and be locally unique. It is not a uniqueness theorem for the authority's full policy problem. In particular, it does not rule out other tax codes, other selected branches, or optimal policies at the boundary of a compliance regime.

\section{The safe--vulnerable composition theorem}
\label{sec:composition}

Let us now consider a reform that raises the capital tax and uses the additional receipts to lower the labor tax. The expenditure requirement is unchanged. At a regular mixing point, the two tax rates can be varied in this way because each raises total revenue at the margin.

There are two reasons to expect the reform to affect the executive's incentives. A higher capital tax directly reduces the return left available for seizure. A lower labor tax, however, changes labor supply and, through production, the return on installed capital. It is not apparent that the first effect should dominate. The following result shows that the comparison is governed by a condition on labor taxation alone.

\begin{theorem}
\label{thm:partial-vulnerability}
At a regular mixing point,
\begin{equation}
 \operatorname{sgn}\left(
 \left.\frac{\dd\log q}{\dd\tau_K}\right|_{R=G}
 \right)
 =
 \operatorname{sgn}\left[1-(1+\varepsilon_L)\tau_L\right].
 \label{eq:composition-sign}
\end{equation}
More precisely,
\begin{equation}
 \left.\frac{\dd\log q}{\dd\tau_K}\right|_{R=G}
 =
 \frac{(1+\varepsilon_K)(1-\theta)
 [1-(1+\varepsilon_L)\tau_L]}
 {\mathcal D(A\zeta+\eta)xy(\cR_L/Y)}.
 \label{eq:composition-level}
\end{equation}
Therefore, if
\begin{equation}
 0<\tau_L<\frac{1}{1+\varepsilon_L},
 \label{eq:labor-threshold}
\end{equation}
a feasible local revenue-neutral increase in the capital-tax rate, accompanied by a
decrease in the labor-tax rate, strictly raises compliance.  At \(\tau_L=0\),
\eqref{eq:composition-level} remains the derivative of the extended revenue contour, but
that direction is blocked by the statutory instrument boundary.
\end{theorem}

The condition has a familiar interpretation. At a fixed wage, a small increase in the labor tax raises labor-tax revenue precisely when \(1-(1+\varepsilon_L)\tau_L>0\). This same expression determines whether shifting revenue toward capital raises compliance. The elasticity of substitution and the capital-supply elasticity affect the size of the response, but do not affect its sign.

This does not mean that general-equilibrium effects can be ignored. They enter the denominator of \eqref{eq:composition-level}, and regularity requires that denominator to be positive. In particular, the fixed-wage labor-revenue condition and the assumption that total equilibrium revenue increases with the labor tax are different restrictions.\footnote{When labor supply is isoelastic, the fixed-wage revenue peak is \(1/(1+\varepsilon_L)\). With general supply costs, \(\varepsilon_L\) is evaluated at the current allocation. The inequality in \eqref{eq:labor-threshold} is then a local condition, not a constant upper bound on admissible labor taxes.}

One can extend the comparison along a segment of the budget contour, provided these restrictions continue to hold.

\begin{corollary}
\label{cor:connected-segment}
On any connected \(C^1\) mixing segment parameterized by \(\chi\), with
\(\dd\tau_K/\dd\chi>0\) and \(\dd\tau_L/\dd\chi<0\), if regularity and
\eqref{eq:labor-threshold} hold pointwise, then \(\dd q/\dd\chi>0\).  The statement
ends at full compliance, certain betrayal, an instrument boundary, a fold, or any failure
of the maintained inequalities.
\end{corollary}

Thus, a sequence of small capitalward reforms raises compliance for as long as the economy remains on the same regular mixing segment. It need not do so after the segment ends. In particular, the theorem does not compare a policy under which the opportunist always betrays with one under which it always complies. We return to the latter regime in Section \ref{sec:allocation}.

\subsection{When the reform also raises investment}

Higher compliance is not, by itself, sufficient for higher investment. Investors receive the after-tax capital return more often, but the statutory claim on that return has also increased. To determine which effect dominates, it is useful to combine the household's capital-supply condition with the executive's indifference condition. This gives an identity that does not depend on the production function.

\begin{proposition}
\label{prop:partial-investment}
Every perturbation that holds \(p\), \(\xi\), \(\beta_O\), and the continuation schedule \(\Delta\) fixed, and remains in the interior mixing regime, satisfies
\begin{equation}
 \dd\log K
 =
 \frac{\varepsilon_K}{1+\varepsilon_K}
 (\zeta-\eta)\dd\log q.
 \label{eq:investment-identity}
\end{equation}
Along a compliance-enhancing reform, installed capital rises if and only if
\(\zeta>\eta\).
\end{proposition}

The two elasticities in \eqref{eq:investment-identity} measure different parts of the adjustment. The first, \(\zeta\), measures how much a rise in compliance improves the expected retention of capital income. The second, \(\eta\), measures how much the resulting deterioration in the posterior reduces the opportunist's continuation gain. Indifference requires the current prize from betrayal to fall with that gain. Investment rises if the improvement in expected retention is sufficiently large relative to this reduction.

When all residual capital income is exposed, \(\zeta=1\), and the condition is simply \(\eta<1\). Partial exposure does not eliminate the mechanism. It makes the condition more demanding, since an increase in compliance then protects a smaller share of the investor's return.

The labor-tax reduction gives a further benefit. If capital does not fall, labor supply rises, as do output and current surplus.

\begin{corollary}
\label{cor:current-surplus}
Suppose the reform in Theorem \ref{thm:partial-vulnerability} lowers \(\tau_L\), raises
\(q\), and satisfies \(0<\tau_L<1/(1+\varepsilon_L)\) and \(\zeta\geq\eta\).  Then
\begin{equation}
 \dd K\geq0,\qquad \dd L>0,\qquad \dd Y>0,
 \label{eq:allocation-signs}
\end{equation}
and
\begin{equation}
 \dd S
 =r[1-m(q)x]\,\dd K+w\tau_L\,\dd L>0.
 \label{eq:surplus-change}
\end{equation}
\end{corollary}

The restriction on the continuation elasticity is substantive. It cannot be replaced by the statement that capital and labor are taxed differently, or by a general appeal to the value of reputation. Section \ref{sec:terminal} gives an explicit continuation economy in which it can be checked from primitives.

\begin{remark}
\label{rem:low-exposure}
For fixed \(q\),
\begin{equation}
 \frac{\partial\zeta(q;\xi)}{\partial\xi}
 =
 \frac{q}{[1-\xi(1-q)]^2}>0.
 \label{eq:zeta-xi}
\end{equation}
Moreover, \(0\leq\zeta(q;\xi)\leq\xi\) for every \(q\in[p,1]\).  Hence
\(\zeta(q;\xi)\to0\) uniformly in \(q\) as \(\xi\downarrow0\).  For a given continuation elasticity, lower exposure therefore makes investment protection harder to obtain from a marginal increase in compliance. This is a comparison at fixed \(q\). Once the equilibrium response of \(q\) is included, neither \(q\) nor \(\zeta\) has been shown to be monotone in exposure.

There is also a limit to how far the mixing calculation can be continued as exposure falls. If
\(xI_K\leq\overline I_K\) on a compact policy branch, then interior mixing requires
\begin{equation}
 \xi\overline I_K
 \geq\xi xI_K
 =\beta_O\Delta(p/q)
 \geq\beta_O\Delta(p).
 \label{eq:mixing-exposure-bound}
\end{equation}
For \(\xi<\beta_O\Delta(p)/\overline I_K\), the current prize is too small to make the opportunist indifferent, so an interior mixing equilibrium is impossible. At zero exposure, betrayal yields no current gain, whereas compliance preserves a strictly positive continuation gain. The opportunist then strictly prefers to comply. Hence the mixing branch must end before exposure reaches zero; it may end at full compliance or at an earlier fold or boundary.
\end{remark}

\section{Ramsey allocations under limited commitment}
\label{sec:allocation}

The preceding section asks what happens to compliance and investment when the tax mix changes within the mixing regime. A different question is how much capital can be supported if the authority insists on full compliance. We now answer this question on a selected branch of the balanced-budget allocation set.

It is useful here to work directly with capital and labor, rather than with their tax rates. Under full compliance, the household's first-order conditions identify the after-tax factor payments. Those payments also determine the amount the executive can seize. The incentive constraint can therefore be written as a restriction on the allocation itself. This is what makes a global comparison possible, although only within the full-compliance regime.

With \(q=1\), household optimality requires
\begin{equation}
 \Psi_K'(K)=(1-\tau_K)F_K(K,L),
 \qquad
 \Psi_L'(L)=(1-\tau_L)F_L(K,L).
 \label{eq:fc-household-focs}
\end{equation}
Multiplying these conditions by the corresponding factor inputs and using constant returns gives the balanced-budget implementability condition
\begin{equation}
 \mathcal H(K,L)
 \equiv
 F(K,L)-K\Psi_K'(K)-L\Psi_L'(L)=G.
 \label{eq:fc-implementability}
\end{equation}

\begin{assumption}
\label{ass:fc-branch}
The selected interior component of \eqref{eq:fc-implementability} is a compact graph
\begin{equation}
 \mathcal A_G^1
 =
 \{(K,\ell_G(K)):K\in[\underline K_G,\overline K_G]\},
 \label{eq:fc-graph}
\end{equation}
where \(\ell_G\) is continuous.  The implied net-of-tax rates lie in \((0,1]\).
Full-compliance surplus
\begin{equation}
 s_G(K)=
 F(K,\ell_G(K))-\Psi_K(K)-\Psi_L(\ell_G(K))
 \label{eq:reduced-surplus}
\end{equation}
is strictly increasing below a unique \(K_G^R\) and strictly decreasing above it.
\end{assumption}

Assumption \ref{ass:fc-branch} isolates the usual single-peaked Ramsey problem on a specified allocation branch. It is not an existence or uniqueness result for the entire balanced-budget set. In particular, it does not rule out other components of that set.

Along the selected branch, neither type betrays. There is therefore no diversion, and observing compliance leaves the posterior equal to \(p\), whatever tax mix is used. The continuation term is constant. The authority's ranking of full-compliance allocations is consequently its ranking of current surplus, \(s_G\). Before imposing the incentive constraint, the preferred allocation has capital \(K_G^R\).

\begin{lemma}
\label{lem:allocation-tax-equivalence}
Every interior full-compliance tax equilibrium on the selected component induces an
allocation in \(\mathcal A_G^1\).  Conversely, conditional on \(q=1\), every
allocation in \(\mathcal A_G^1\) is implemented as a competitive and fiscally balanced
allocation by the unique code
\begin{equation}
 \tau_K(K)
 =
 1-\frac{\Psi_K'(K)}{F_K(K,\ell_G(K))},
 \qquad
 \tau_L(K)
 =
 1-\frac{\Psi_L'(\ell_G(K))}
 {F_L(K,\ell_G(K))}.
 \label{eq:implementing-taxes}
\end{equation}
It is supported by a full-compliance tax-code equilibrium if and only if it also satisfies \eqref{eq:allocation-ic}.
\end{lemma}

The remaining constraint concerns the executive. Let
\begin{equation}
 \phi_K(K)=K\Psi_K'(K).
 \label{eq:capital-payment-map}
\end{equation}
This is the household's total after-tax capital payment when compliance is certain. It is strictly increasing:
\begin{equation}
 \phi_K'(K)=\Psi_K'(K)+K\Psi_K''(K)>0.
 \label{eq:phi-monotone}
\end{equation}
A fraction \(\xi\) of this payment can be seized. The opportunist therefore complies if and only if
\begin{equation}
 \xi\phi_K(K)\leq\beta_O\Delta(p).
 \label{eq:allocation-ic}
\end{equation}
The right-hand side is the discounted gain from preserving the prior reputation. Thus, limited commitment places an upper bound on the capital payment, and hence on capital itself. To express this bound on the selected allocation branch, write
\begin{equation}
 b(p;\beta_O,\xi)=\frac{\beta_O\Delta(p)}{\xi},
 \qquad
 \mathcal K_G^D(p)
 =
 \left\{
 K\in[\underline K_G,\overline K_G]:
 \phi_K(K)\leq b(p;\beta_O,\xi)
 \right\}.
 \label{eq:credibility-capacity}
\end{equation}
If this set is nonempty, let its largest element be
\begin{equation}
 K_G^D(p)=\max\mathcal K_G^D(p).
 \label{eq:capped-credibility-ceiling}
\end{equation}
If \(b\) lies between \(\phi_K(\underline K_G)\) and
\(\phi_K(\overline K_G)\), then
\(K_G^D=\phi_K^{-1}(b)\).  If
\(b\geq\phi_K(\overline K_G)\), then \(K_G^D=\overline K_G\).

\begin{proposition}
\label{prop:allocation-ceiling}
Maintain Assumption \ref{ass:fc-branch}.  Then:
\begin{enumerate}[label=(\roman*)]
 \item a full-compliance allocation on the selected component exists if and only if
 \[
  b(p;\beta_O,\xi)\geq\phi_K(\underline K_G);
 \]
 \item conditional on existence, the unique constrained Ramsey allocation is
 \begin{equation}
  K_G^C(p)=\min\{K_G^R,K_G^D(p)\},
  \qquad
  L_G^C(p)=\ell_G(K_G^C(p));
  \label{eq:allocation-ceiling-rule}
 \end{equation}
 \item the unconstrained Ramsey allocation is credible if and only if
 \[
  b(p;\beta_O,\xi)\geq\phi_K(K_G^R).
 \]
 Otherwise the credibility constraint binds and
 \(K_G^C(p)=K_G^D(p)<K_G^R\).
\end{enumerate}
\end{proposition}

To see the result, suppose first that the ordinary Ramsey allocation gives the executive no reason to betray. There is then no reason to change it. If instead its capital payment is too large, capital must be reduced to the credibility ceiling. Reducing it further would lower surplus without being required for compliance, since surplus is increasing up to the ordinary Ramsey allocation.

This conclusion is conditional on choosing full compliance. A policy with mixing may still be preferable, and Proposition \ref{prop:allocation-ceiling} does not rule it out. We also use the convention that the opportunist complies when indifferent.\footnote{At equality, compliance is a best response but not a strict one. The allocation is therefore weakly implementable. Indifference at the full-compliance allocation does not, on its own, establish another rational-expectations equilibrium with mixing: a different compliance probability changes household choices and the current prize.}

Better reputation or a more patient opportunist raises the ceiling, while greater exposure lowers it. These comparisons are immediate once the constraint is written in allocation space, but it matters which continuation objects are held fixed.

\begin{corollary}
\label{cor:allocation-comparative-statics}
Suppose \(\Delta\) is strictly increasing.  On the interior binding region
\[
 \underline K_G<K_G^C=K_G^D<K_G^R,
\]
holding the continuation schedule fixed,
\begin{align}
 \frac{\partial K_G^C}{\partial p}
 &=
 \frac{\beta_O\Delta'(p)}
 {\xi\phi_K'(K_G^C)}>0,
 &
 \frac{\partial K_G^C}{\partial\beta_O}
 &=
 \frac{\Delta(p)}
 {\xi\phi_K'(K_G^C)}>0,
 \label{eq:ceiling-positive}\\
 \frac{\partial K_G^C}{\partial\xi}
 &=
 -\frac{\beta_O\Delta(p)}
 {\xi^2\phi_K'(K_G^C)}<0.
 \label{eq:ceiling-exposure}
\end{align}
Under strict slack, \(b>\phi_K(K_G^R)\), these derivatives are locally zero.  At the
feasibility boundary and at \(b=\phi_K(K_G^R)\), the corresponding statements are
one-sided.  Conditional on full-compliance feasibility, constrained surplus is weakly
increasing in \(p\) and \(\beta_O\), and weakly decreasing in \(\xi\), with strict
changes whenever the ceiling lies strictly between \(\underline K_G\) and \(K_G^R\).
\end{corollary}

The comparative statics concern the same continuation schedule. If changing the opportunist's discount factor also changes \(\Delta\), the numerator of the derivative with respect to \(\beta_O\) becomes \(\Delta+\beta_O\Delta_{\beta_O}\). The sign then requires an additional restriction. Section \ref{sec:dynamics} returns to this distinction when the authority and the executive become more patient together.

One might expect the capital ceiling to translate into a lower bound on the capital-tax rate. Such a translation is not generally available. A change in the allocation also changes the marginal product of capital, and the implementing tax rate depends on both.

\begin{remark}
\label{rem:no-tax-floor}
At every point where the allocation branch is differentiable, the implementing rate
obeys
\begin{equation}
 \frac{\dd\tau_K(K)}{\dd K}
 =
 -\frac{\Psi_K'(K)}{F_K}
 \left[
 \frac{\Psi_K''(K)}{\Psi_K'(K)}
 -
 \frac{F_{KK}+F_{KL}\ell_G'(K)}{F_K}
 \right].
 \label{eq:tax-implementation-slope}
\end{equation}
The maintained assumptions do not sign the bracket along an arbitrary selected branch. A tighter binding constraint reduces capital, but the induced change in the statutory capital-tax rate also depends on how labor and the marginal product of capital adjust. The allocation ceiling should therefore not be described as a general tax-rate floor.
\end{remark}

\section{A primitive continuation benchmark}
\label{sec:terminal}

So far, we have treated the opportunist's continuation gain as a function of reputation and stated the investment result in terms of its elasticity. This leaves an important question: can the required elasticity arise in an ordinary production economy? The following two-period example gives an affirmative, but qualified, answer. It also illustrates why the curvature of the authority's continuation value must be examined separately.

We use full exposure in the terminal period. This choice makes the continuation value particularly simple: reputation enters the expected return on capital multiplicatively. With partial terminal exposure, expected retention is instead \(1-\xi(1-z)\), and the power formula derived below no longer applies.

The tax-composition problem of the preceding sections takes place in period one. In period two, households invest in a new Cobb--Douglas production opportunity,
\begin{equation}
 Y_2=\mathcal A_2K_2^\alpha L_2^{1-\alpha},
 \qquad \alpha\in(0,1),
 \label{eq:terminal-production}
\end{equation}
under a tax package \((\bar\tau_K,\bar\tau_L)\) fixed before the period-one signal. Write \(\bar x=1-\bar\tau_K\) and \(\bar y=1-\bar\tau_L\). Terminal factor-supply costs are isoelastic, with elasticities \(\bar\varepsilon_K,\bar\varepsilon_L>0\). There is no future reputation to preserve after period two, so an opportunist seizes the entire after-tax capital return. A capital investor therefore receives that return only when the executive is committed, an event with probability \(z\).

The fixed terminal code is a continuation experiment, not a second solution of the exact-budget Ramsey problem. Its receipts vary with the inherited reputation; we do not require it to finance the same fixed \(G\) at every posterior.\footnote{Ordinary terminal tax receipts can be returned through an account that does not affect atomistic marginal choices. Fixing the terminal code allows us to isolate the effect of reputation on investment. Requiring a new exact expenditure balance and allowing taxes to be reset would introduce a further policy problem, whose solution need not have the same elasticity.}

The terminal household conditions can be solved jointly with production. Output takes the form
\begin{equation}
 Y_2(z)=\bar Yz^\gamma,
 \qquad
 \gamma=
 \frac{\alpha\bar\varepsilon_K(1+\bar\varepsilon_L)}
 {1+(1-\alpha)\bar\varepsilon_K+\alpha\bar\varepsilon_L}.
 \label{eq:terminal-gamma}
\end{equation}

Thus, \(\gamma\) is the elasticity of terminal output with respect to reputation, including the adjustment of both factors. The opportunist's terminal payoff is the residual capital income it seizes,
\begin{equation}
 V_O(z)=B_2z^\gamma,
 \qquad
 B_2=\bar x\alpha\bar Y>0.
 \label{eq:terminal-opportunist}
\end{equation}
Because output vanishes at zero reputation, \(V_O(0)=0\). It follows that
\begin{equation}
 \Delta(z)=B_2z^\gamma,\qquad \eta(z)=\gamma.
 \label{eq:terminal-delta}
\end{equation}

\begin{proposition}
\label{prop:primitive-eta}
In the terminal benchmark,
\begin{equation}
 \eta\leq1
 \quad\Longleftrightarrow\quad
 \alpha\bar\varepsilon_K(1+\bar\varepsilon_L)
 \leq
 1+(1-\alpha)\bar\varepsilon_K+\alpha\bar\varepsilon_L.
 \label{eq:gamma-condition}
\end{equation}
A sufficient condition is
\begin{equation}
 \alpha\leq\frac12,\qquad \bar\varepsilon_K\leq1.
 \label{eq:simple-gamma-condition}
\end{equation}
No restriction on \(\bar\varepsilon_L\) is then required.
\end{proposition}

The condition is not automatic. It restricts how strongly future investment and output respond to reputation. When it holds, the full-exposure version of the investment theorem applies; with partial exposure in period one, the relevant comparison is still the stricter \(\zeta\geq\gamma\).

We can also derive the authority's continuation value in this example. Treat terminal seizure as a transfer, so that terminal welfare is output less the two supply costs. Resource surplus is then
\begin{equation}
 S_2(z)
 =
 \bar Y\left[(1-\ell)z^\gamma-kz^{\gamma+1}\right],
 \label{eq:terminal-surplus}
\end{equation}
where
\begin{equation}
 \ell=
 \frac{\bar\varepsilon_L}{1+\bar\varepsilon_L}
 \bar y(1-\alpha),
 \qquad
 k=
 \frac{\bar\varepsilon_K}{1+\bar\varepsilon_K}
 \bar x\alpha.
 \label{eq:terminal-cost-shares}
\end{equation}

\begin{proposition}
\label{prop:terminal-welfare}
If \(0<\gamma\leq1\), then \(S_2(0)=0\), \(S_2\) is strictly concave on
\((0,1]\), and
\begin{equation}
 \frac{\dd}{\dd q}\left[qS_2(p/q)\right]
 =
 S_2(z)-zS_2'(z)>0.
 \label{eq:terminal-smoothing}
\end{equation}
Consequently, a reform satisfying the conditions of Corollary \ref{cor:current-surplus}, with \(\eta=\gamma\) and \(0<\gamma\leq1\), raises both current resource surplus and expected terminal resource surplus in this fixed-terminal-policy benchmark.
\end{proposition}

The last result has a simple interpretation. Greater compliance makes the terminal reputation less dispersed without changing its mean. Since terminal surplus is concave in reputation here, this reduction in dispersion is beneficial. The example therefore need not generate a conflict between current investment and future welfare: both can improve under the same reform.

This observation also indicates why an informational cost of compliance cannot be assumed throughout the paper. Such a cost arises when continuation value is convex, whereas the continuation value in this example is concave. The next section treats the two possibilities separately. Nor does the example establish a monotone policy response to a common discount factor: changing that factor changes current compliance as well as the weight on the future.

\section{Reputation and the dynamic composition margin}
\label{sec:dynamics}

A tax code also affects how much is learned about the executive. Suppose that a reform makes compliance more likely. This reduces the probability of betrayal, which reveals the opportunistic type. It also makes compliance less reassuring: an executive that faced a smaller temptation has passed a less demanding test.

These two changes leave mean reputation unchanged. The posterior is zero after betrayal and \(p/q\) after compliance, with probabilities \(1-q\) and \(q\), respectively. Its mean is therefore always \(p\). Raising \(q\) moves the posterior distribution toward that mean. Whether the authority values this change depends on the curvature of its continuation payoff, not on Bayes' rule alone.

\begin{proposition}
\label{prop:information}
Along an interior mixing branch, holding \(p\), \(\beta_O\), and the continuation schedules fixed,
\begin{equation}
 \left.\frac{\dd D}{\dd q}\right|_{\mathrm{IC}}
 =
 -X\left[1+\eta\frac{1-q}{q}\right]<0,
 \label{eq:diversion-derivative}
\end{equation}
and
\begin{equation}
 \cJ_q(q;p)
 =
 \cV(z)-z\cV'(z)-\cV(0),
 \qquad z=p/q.
 \label{eq:Jq}
\end{equation}
Thus \(\cJ_q\leq0\) if \(\cV\) is convex and
\(\cJ_q\geq0\) if \(\cV\) is concave, with strict inequalities under strict
curvature.
\end{proposition}

There are accordingly two distinct effects. Expected diversion falls because betrayal occurs less often and the amount seized when it occurs is smaller. The informational effect can have either sign. With convex continuation value, the authority would prefer a more dispersed posterior distribution, and the reduction in learning is a cost. With concave continuation value, as in the preceding example, the same change is a benefit.

This distinction matters for policy comparisons. Greater concern for future outcomes need not always favor a less protective tax code. We first give a monotonicity result for the convex case, holding the executive's incentives fixed. We then ask what remains when patience changes those incentives as well.

\subsection{A continuation-weight comparative static}

Consider a compact connected regular segment \(X_G\) of the balanced-budget policy set. Index it by \(\chi\), with larger \(\chi\) corresponding to more capital taxation and less labor taxation. Hold the opportunist's discount factor \(\beta_O\) and the continuation schedules fixed, so that the compliance map \(q_G(\chi)\) does not change in the comparison. Let
\begin{equation}
 c_G(\chi)=S(\chi)-\lambda_DD(\chi),
 \qquad
 j_G(\chi)=\cJ(q_G(\chi);p).
 \label{eq:reduced-components}
\end{equation}
Changing only the authority's continuation weight gives the problem
\begin{equation}
 \max_{\chi\in X_G}\;
 W_G(\chi;\omega),
 \qquad
 W_G(\chi;\omega)=c_G(\chi)+\omega j_G(\chi).
 \label{eq:fixed-weight-problem}
\end{equation}

\begin{proposition}
\label{prop:continuation-weight}
Suppose \(c_G,j_G\) are continuous and \(j_G\) is weakly decreasing.  The smallest
and largest optimal elements of \eqref{eq:fixed-weight-problem} are weakly decreasing
in \(\omega\).  If \(j_G\) is strictly decreasing between distinct candidate optima,
every optimal selection is weakly decreasing.
\end{proposition}

On the segment described by Theorem \ref{thm:partial-vulnerability}, compliance rises with \(\chi\). If \(\cV\) is convex, expected continuation value therefore falls with \(\chi\), and the proposition applies. Giving the future more weight moves the smallest and largest optimal policies weakly away from capital taxation.

The argument does not require a concave objective or an interior solution. It is a comparison of the rankings of any two feasible policies: the more capital-heavy policy has the lower continuation payoff, so raising the weight on that payoff cannot make it relatively more attractive. What is essential is that the feasible set and the executive's response to each policy remain fixed. With concave continuation value the ordering of continuation payoffs is reversed, and this particular argument no longer gives a shift away from capital.

\subsection{Common patience}

Suppose instead that the authority and the opportunist share a discount factor \(\beta\). An increase in patience then does two things. It makes future information more valuable to the authority, but it also makes betrayal less attractive to the executive. A given tax code no longer induces the same compliance probability. The preceding result cannot be applied without accounting for that response.

To separate these effects, initially keep \(\Delta\) and \(\cV\) fixed as functions of the posterior, and assume that current welfare has no direct dependence on \(\beta\).\footnote{This is the common-discount comparison in a two-period economy with fixed terminal primitives. In a stationary economy, the continuation schedules would generally change with the discount factor as well. We discuss that further dependence below; we do not solve a stationary comparative-static problem here.}

Write \(b=\log\beta\), and use exact balance to express the local policy problem in terms of \(\chi\) and \(q\). The mixing condition is
\begin{equation}
 X(\chi,q)=\beta\Delta(p/q).
 \label{eq:common-beta-ic}
\end{equation}
For this reduced problem, define
\begin{equation}
 a=\frac{\partial\log X}{\partial\log q},
 \qquad
 g=-\frac{\partial\log X}{\partial\chi}>0,
 \qquad
 H=a+\eta>0.
 \label{eq:common-beta-elasticities}
\end{equation}
The first inequality says that a more capital-heavy policy lowers the prize at a given compliance probability; the second is the local stability condition for the mixing response. Differentiating the incentive condition gives
\begin{equation}
 q_\chi=\frac{qg}{H}>0,
 \qquad
 q_b=\frac qH>0.
 \label{eq:common-beta-q}
\end{equation}
Thus, greater patience raises compliance even before the authority changes its policy. The marginal compliance effect of changing the tax mix may change as well.

Let \(C(\chi,q)\) denote current welfare after the other static variables have been substituted out, and write \(J(q)=\cJ(q;p)\). The authority now maximizes
\begin{equation}
 \Pi(\chi,b)
 =
 C(\chi,q(\chi,b))+\beta J(q(\chi,b)).
 \label{eq:common-beta-objective}
\end{equation}

\begin{proposition}
\label{prop:common-beta}
At a twice differentiable interior mixing point,
\begin{align}
 \Pi_{\chi b}
 &=
 \beta J_q q_\chi
 +q_b\left[
 C_{\chi q}+(C_{qq}+\beta J_{qq})q_\chi
 \right]
 +(C_q+\beta J_q)q_{\chi b}.
 \label{eq:common-beta-cross}
\end{align}
At a unique interior optimum satisfying the strict second-order condition \(\Pi_{\chi\chi}<0\),
\begin{equation}
 \operatorname{sgn}\frac{\dd\chi^*}{\dd b}
 =
 \operatorname{sgn}\Pi_{\chi b}.
 \label{eq:common-beta-optimum}
\end{equation}
In particular, a sufficient set of conditions for
\(\Pi_{\chi b}\leq0\) is
\begin{align}
 J_q&\leq0,
 \label{eq:common-sufficient-one}\\
 C_{\chi q}+(C_{qq}+\beta J_{qq})q_\chi&\leq0,
 \label{eq:common-sufficient-two}\\
 (C_q+\beta J_q)q_{\chi b}&\leq0.
 \label{eq:common-sufficient-three}
\end{align}
Suppose the same compact ordered policy set \(X\) is available for every \(b\) in an
interval, after a \(b\)-independent parameterization, and \(\Pi\) is continuous on the
product set.  If, for every \(\chi_2>\chi_1\) and \(b_2>b_1\),
\begin{align}
 &[\Pi(\chi_2,b_2)-\Pi(\chi_1,b_2)]
 -[\Pi(\chi_2,b_1)-\Pi(\chi_1,b_1)]
 \leq0,
 \label{eq:common-global-dd}
\end{align}
then the smallest and largest optimal policies are weakly decreasing in common
patience.  On smooth cells, \(\Pi_{\chi b}\leq0\), together with the corresponding
finite-difference inequalities across regime boundaries and kinks, is sufficient.
\end{proposition}

The first term in \eqref{eq:common-beta-cross} is the effect retained in the fixed-weight comparison. Under convex continuation value, it favors moving away from the vulnerable base. The second term accounts for the change in the marginal value of a reform when patience raises compliance. The third accounts for the change in the compliance response to the reform itself.

Neither the labor-tax condition in Theorem \ref{thm:partial-vulnerability} nor the investment condition \(\zeta\geq\eta\) signs the last two terms. Hence they are not sufficient for a common-patience theorem. The next example shows the failure of the inference in a reduced policy problem. Its purpose is limited: it is not a counterexample constructed from the full production economy.

\begin{example}
\label{ex:common-beta-counterexample}
Let
\[
 p=\frac1{100},\qquad
 \chi\in\left[\frac1{10},\frac12\right],\qquad
 \beta\in\left[\frac3{20},\frac3{10}\right],
\]
and set
\begin{equation}
 \Delta(z)=z,
 \qquad
 X(\chi,q)=\frac1{100\chi}.
 \label{eq:counterexample-primitives}
\end{equation}
The mixing condition gives \(q=\beta\chi\), which lies strictly between \(p\) and one.
Let current welfare be \(C=q\), and let
\[
 \cV(z)=100z^2,\qquad \cV(0)=0.
\]
Then continuation value is strictly convex and
\[
 J(q)=q\cV(p/q)=\frac1{100q}.
\]
A more capital-heavy policy therefore raises current welfare and lowers expected continuation welfare. These are the same rankings used in the fixed-weight argument. When the discount factor is common, however, substituting the compliance response gives
\begin{equation}
 \Pi(\chi,\beta)
 =\beta\chi+\frac1{100\chi}.
 \label{eq:counterexample-objective}
\end{equation}
The objective is strictly convex in \(\chi\), so it suffices to compare the two endpoints:
\begin{equation}
 \Pi(1/2,\beta)-\Pi(1/10,\beta)
 =\frac25\left(\beta-\frac15\right).
 \label{eq:counterexample-switch}
\end{equation}
Thus
\begin{equation}
 \argmax_\chi\Pi(\chi,\beta)
 =
 \begin{cases}
  \{1/10\},&\beta<1/5,\\
  \{1/10,1/2\},&\beta=1/5,\\
  \{1/2\},&\beta>1/5.
 \end{cases}
 \label{eq:counterexample-argmax}
\end{equation}
As common patience crosses \(1/5\), the authority switches toward the vulnerable base. By contrast, if the opportunist's discount factor were held fixed and only the social continuation weight changed, Proposition \ref{prop:continuation-weight} would apply.
\end{example}

The reason for the reversal is visible in \eqref{eq:counterexample-objective}. The direct increase in the continuation weight is exactly offset by the change in compliance inside the continuation payoff. Meanwhile, current compliance becomes more responsive to the policy. Increasing common patience then favors the capital-heavy endpoint. The example establishes that the fixed-weight assumptions are insufficient; it does not establish that the same reversal must occur under any particular production technology.

In a stationary model there is another adjustment to consider. If the opportunist's continuation schedule depends directly on \(b\), let
\begin{equation}
 \nu_\Delta(z,b)
 =
 \left.\partial_b\log\Delta(z,b)\right|_z.
 \label{eq:nu-delta}
\end{equation}
Then
\begin{equation}
 q_b=\frac{q[1+\nu_\Delta(z,b)]}{H}.
 \label{eq:stationary-qb}
\end{equation}
If the authority's continuation payoff is also \(J=J(q,b)\), the first term in \eqref{eq:common-beta-cross} becomes \(\beta(J_q+J_{qb})q_\chi\), where \(J_{qb}\) is the direct derivative at fixed \(q\). The other terms keep their form, using the modified derivatives of \(q\), provided current welfare has no direct \(b\)-dependence.

A stationary monotonicity result would therefore require more than the local production theorem. One would have to solve for the continuation schedules and show that the relevant equilibrium selection is differentiable and preserves the required ordering of policies. We make no such claim here.

\section{Exposure, institutional scope, and interpretation}
\label{sec:scope}

The argument relies on a difference in what remains exposed after private decisions have been made. It does not rely on capital and labor having intrinsically different claims to protection. We conclude the analysis by separating this exposure difference from two other features of the model: what the executive can seize, and who chooses the ordinary tax code.

\subsection{If both factor payments are exposed}
Suppose that betrayal allows the executive to seize fractions \(\xi_K\) and \(\xi_L\) of the after-tax capital and labor payments. The prize can then be written as
\begin{equation}
 X_{\boldsymbol\xi}
 =
 \xi_KxrK+\xi_LywL
 =
 \xi_L(Y-G)+(\xi_K-\xi_L)xrK.
 \label{eq:exposure-decomposition}
\end{equation}
The second equality follows from constant returns and exact budget balance. It separates the exposed share of total private income from the additional exposure of the capital payment. At a given output level, shifting the tax mix changes the prize through the latter term, whose coefficient is \(\xi_K-\xi_L\).

If the two payments are equally exposed, this direct composition effect disappears. A change in taxes can still affect temptation by changing output, but it no longer pre-collects one payment that is more exposed than the other. This distinction identifies the source of the mechanism; it does not extend the exact sign formula to two vulnerable bases. If labor income can also be seized, its expected return depends on compliance, and the factor-supply calculation must be done again.

At full compliance with both payments exposed, the allocation-level prize becomes
\begin{equation}
 \Phi_G(K)
 =
 \xi_KK\Psi_K'(K)
 +\xi_L\ell_G(K)\Psi_L'(\ell_G(K)).
 \label{eq:two-base-allocation-prize}
\end{equation}
If \(\Phi_G\) is strictly increasing, the same truncation argument used in Proposition \ref{prop:allocation-ceiling} applies. With two exposed payments, however, this monotonicity is an additional restriction. Labor may fall as capital rises along the balanced-budget branch, so the labor payment can offset part of the increase in the capital payment.

\subsection{What the executive can seize}
In the model, the executive seizes a flow of income after capital has been installed. The closest applications are therefore projects, concessions, or investments whose residual cash flow remains exposed after the investor has incurred the relevant cost. The interpretation is about the timing and enforceability of the payment, not about whether the investor is domestic or foreign.

If the executive can seize the physical asset or its undepreciated principal, ordinary income taxation need not reduce the prize in the same way. An extension to durable capital would have to account for the value of the remaining asset. The present theorem should not be applied to that environment by simply renaming the income flow.

\subsection{Who chooses the ordinary tax code}
The authority that chooses ordinary taxes does not know the executive's type. This assumption allows us to study a Ramsey choice of tax composition without also making the tax code a signal of type. If the informed executive chose the code itself, alternative announcements would require beliefs, and the equilibrium-selection issue emphasized by \citet{Lu2013} would reappear, here with more than one instrument.

The institutional separation is therefore part of the economics, not merely a way of simplifying the solution. Ordinary revenue must be committed to the expenditure requirement, while the executive retains a subsequent power to interfere with residual private payments. One can think of appropriated receipts being protected by a fiscal procedure that does not fully protect property rights after investment. The model is informative to the extent that this asymmetry describes the institution under consideration.

None of these qualifications makes capital taxation nondistortionary. Holding compliance and the labor tax fixed, a higher capital tax lowers investment. The protective effect arises because compliance changes when the executive's temptation changes. Our results describe this adjustment on a regular mixing segment, and the separate allocation restriction needed to sustain full compliance. They do not rank disconnected fiscal branches or provide a general policy theorem for a common discount factor.

\section{Conclusion}

Can a tax on capital protect investment? In the environment studied here, it can. An investor cares about the ordinary tax bill, but also about what the executive may take after the investment has been made. Collecting more of an exposed return through ordinary taxation leaves a smaller prize for a subsequent seizure. When this makes compliance sufficiently more likely, the investor can be better off despite the higher statutory rate.

The paper gives conditions under which this argument survives the joint adjustment of capital, labor, factor prices, and tax revenue. On a regular mixing segment, a revenue-neutral shift from labor to capital taxation raises compliance if the labor tax is below its local fixed-wage revenue peak. The sign does not depend on the elasticity of substitution. Investment also rises when expected retention responds more strongly than the opportunist's continuation gain. Under that additional condition, the accompanying labor-tax reduction raises labor, output, and current surplus.

There is a related restriction on what can be achieved without any betrayal. At full compliance, the executive's temptation is determined by the exposed after-tax capital payment. Credibility places a ceiling on that payment and, consequently, on the capital allocation. On a single-peaked Ramsey branch, the best full-compliance allocation is the ordinary Ramsey choice truncated at this ceiling. The corresponding statutory tax rate need not move monotonically, since factor prices and labor supply also change.

Finally, protecting investment affects learning about the executive. Greater compliance makes the posterior less dispersed while preserving its mean. This is costly under convex continuation value and beneficial under concave continuation value. Even in the convex case, a policy comparison based on giving the authority more weight on the future must be distinguished from one in which both the authority and the executive become more patient. The latter also changes the executive's incentives.

The broader implication is that optimal tax composition depends on what can still be taken after private choices are sunk. Ordinary tax distortions remain, but they are not the only relevant margin when private income is unevenly exposed to a later exercise of government power.

\clearpage
\appendix

\section{Proofs for the local composition results}
\label{app:local}

We first establish the local equilibrium response and then impose revenue neutrality. Keeping these steps separate is useful: the positivity of a fixed-tax compliance response does not by itself determine the response to a change in tax composition.

\begin{proof}[Proof of Proposition \ref{prop:mixing-existence}]
By continuity and the strict endpoint inequalities, \(\Phi\) has a zero in \((p,1)\). To establish uniqueness, differentiate the log gap while allowing the household allocation to adjust:
\[
 \frac{\dd\Phi(q)}{\dd\log q}=A(q)\zeta(q)+\eta(p/q)>0,
\]
Thus \(\Phi\) is strictly increasing, and there can be only one zero. Its nonzero derivative permits the use of the implicit-function theorem for local changes in the tax code and prior. At an exactly balanced solution, \(\cR_L>0\) also permits solving the budget equation locally for \(\tau_L\) as a function of \(\tau_K\). The slope is \(-\cR_K/\cR_L<0\), giving the required orientation of the budget contour.
\end{proof}

\begin{proof}[Proof of Lemma \ref{lem:factor-responses}]
Start from the household conditions \eqref{eq:household-focs}. Evaluating the supply elasticities at the current allocation and taking proportional changes gives
\begin{equation}
 \wh K=\varepsilon_K(\wh u+\wh r),
 \qquad
 \wh L=\varepsilon_L(\wh v+\wh w).
 \label{eq:A-factor-focs}
\end{equation}
Now substitute the factor-price responses in \eqref{eq:factor-price-responses}. The resulting system is
\begin{align}
 \left[1+\frac{\varepsilon_K(1-\theta)}{\sigma}\right]\wh K
 -\frac{\varepsilon_K(1-\theta)}{\sigma}\wh L
 &=\varepsilon_K\wh u,
 \label{eq:A-factor-K}\\
 -\frac{\varepsilon_L\theta}{\sigma}\wh K
 +\left[1+\frac{\varepsilon_L\theta}{\sigma}\right]\wh L
 &=\varepsilon_L\wh v.
 \label{eq:A-factor-L}
\end{align}
The determinant of the coefficient matrix is \(\mathcal D>0\), so the system is invertible. Inverting it gives \eqref{eq:factor-matrix}--\eqref{eq:factor-coefficients}; taking the determinant of that response matrix gives \(cf-de=\varepsilon_K\varepsilon_L/\mathcal D\).

For gross factor payments, it is convenient to use household optimality once more:
\begin{equation}
 \wh I_K
 =
 \left(1+\frac1{\varepsilon_K}\right)\wh K-\wh u,
 \qquad
 \wh I_L
 =
 \left(1+\frac1{\varepsilon_L}\right)\wh L-\wh v.
 \label{eq:A-income}
\end{equation}
Substituting the responses of \(K\) and \(L\) into these expressions gives \eqref{eq:income-responses}--\eqref{eq:ABCE}. Since \(c,d>0\), we have \(A+1>0\) and \(B>0\). Subtracting one in the definition of \(A\) prevents the same argument from signing \(A\) itself.
\end{proof}

\begin{proof}[Proof of Theorem \ref{thm:partial-vulnerability}]
We calculate the revenue response to each tax separately and then combine the two changes so that total revenue is unchanged. If \(\tau_L\) is held fixed, the compliance equation \eqref{eq:compliance-response} and \(\wh u=\wh x+\zeta\wh q\) give
\begin{equation}
 \wh q
 =
 \frac{A+1}{(A\zeta+\eta)x}\,\dd\tau_K,
 \qquad
 \wh u
 =
 \frac{\zeta-\eta}{(A\zeta+\eta)x}\,\dd\tau_K.
 \label{eq:A-K-perturbation}
\end{equation}
The capital-tax revenue derivative includes both the direct increase in the rate and the changes in the two factor payments:
\begin{equation}
 \frac{\cR_K}{Y}
 =
 \theta+
 \frac{\zeta-\eta}{(A\zeta+\eta)x}
 [\tau_K\theta A+\tau_L(1-\theta)C].
 \label{eq:A-RK}
\end{equation}
For a change in \(\tau_L\) with \(\tau_K\) held fixed, the corresponding responses are
\begin{equation}
 \wh q=\frac{B}{(A\zeta+\eta)y}\,\dd\tau_L,
 \quad
 \wh u=\frac{\zeta B}{(A\zeta+\eta)y}\,\dd\tau_L,
 \quad
 \wh v=-\frac{\dd\tau_L}{y}.
 \label{eq:A-L-perturbation}
\end{equation}
Substituting them into the revenue equation gives
\begin{equation}
 \frac{\cR_L}{Y}
 =
 (1-\theta)
 -\frac{\tau_K\theta B\eta}{(A\zeta+\eta)y}
 +\frac{\tau_L(1-\theta)
 [C\zeta B-E(A\zeta+\eta)]}
 {(A\zeta+\eta)y}.
 \label{eq:A-RL}
\end{equation}

We can now choose the labor-tax change that offsets the capital-tax revenue change. Exact balance requires
\[
 \frac{\dd\tau_L}{\dd\tau_K}=-\frac{\cR_K}{\cR_L}.
\]
Using this slope in \eqref{eq:compliance-response}, we obtain
\begin{equation}
 (A\zeta+\eta)
 \left.\frac{\dd\log q}{\dd\tau_K}\right|_{R=G}
 =
 \frac{A+1}{x}
 -\frac{B}{y}\frac{\cR_K}{\cR_L}.
 \label{eq:A-contour-q}
\end{equation}
Since \(x,y,\cR_L/Y\) are positive, multiplying the right-hand side by \(xy(\cR_L/Y)\) preserves its sign. Denote the resulting numerator by
\begin{equation}
 Q
 =
 y(A+1)\frac{\cR_L}{Y}
 -xB\frac{\cR_K}{Y}.
 \label{eq:A-Q}
\end{equation}
Substitute \eqref{eq:A-RK}--\eqref{eq:A-RL} and use \(x+\tau_K=y+\tau_L=1\). The terms containing the continuation and retention elasticities cancel, leaving
\begin{equation}
 Q
 =
 (1-\theta)(A+1)-\theta B
 +\tau_L(1-\theta)[BC-(A+1)(E+1)].
 \label{eq:A-Q-collected}
\end{equation}
There are two remaining coefficient combinations. The first follows directly from the definitions in Lemma \ref{lem:factor-responses}; the second uses the determinant identity \(cf-de=\varepsilon_K\varepsilon_L/\mathcal D\):
\begin{align}
 (1-\theta)(A+1)-\theta B
 &=
 \frac{(1+\varepsilon_K)(1-\theta)}{\mathcal D},
 \label{eq:A-identity-one}\\
 BC-(A+1)(E+1)
 &=
 -\frac{(1+\varepsilon_K)(1+\varepsilon_L)}{\mathcal D}.
 \label{eq:A-identity-two}
\end{align}
Hence
\begin{equation}
 Q
 =
 \frac{(1+\varepsilon_K)(1-\theta)}{\mathcal D}
 [1-(1+\varepsilon_L)\tau_L].
 \label{eq:A-Q-final}
\end{equation}
Combining this expression with \eqref{eq:A-contour-q} gives the level formula \eqref{eq:composition-level}. All factors outside the square bracket are positive at a regular mixing point, which proves the sign result. If the labor tax is positive, the implied decrease in that tax is locally feasible. At zero, the derivative still describes the extended contour, but the statutory constraint prevents moving in that direction.
\end{proof}

\begin{proof}[Proof of Corollary \ref{cor:connected-segment}]
Apply Theorem \ref{thm:partial-vulnerability} at each point of the maintained \(C^1\) segment. Since \(\tau_K\) increases with \(\chi\), the derivative of \(q\) with respect to \(\chi\) is strictly positive. Integrating between any two points proves the asserted ordering. The argument applies only while the path remains in the stated regime and satisfies the maintained inequalities.
\end{proof}

\begin{proof}[Proof of Proposition \ref{prop:partial-investment}]
The production terms can be removed from the prize by using capital-supply optimality:
\[
 xr=\frac{\Psi_K'(K)}{m(q)}.
\]
Thus
\begin{equation}
 X=\xi xrK=\frac{\xi K\Psi_K'(K)}{m(q)}.
 \label{eq:A-prize-capital}
\end{equation}
Holding exposure fixed and taking proportional changes gives
\begin{equation}
 \wh X
 =
 \left(1+\frac1{\varepsilon_K}\right)\wh K-\zeta\wh q.
 \label{eq:A-prize-capital-diff}
\end{equation}
With \(p\), \(\beta_O\), and \(\Delta\) fixed, the mixing condition also gives \(\wh X=-\eta\wh q\). Equating the two expressions for \(\wh X\) and solving for \(\wh K\) proves \eqref{eq:investment-identity}.
\end{proof}

\begin{proof}[Proof of Corollary \ref{cor:current-surplus}]
Proposition \ref{prop:partial-investment} first gives \(\wh K\geq0\). To determine the labor response, combine labor-supply optimality with the wage response in \eqref{eq:factor-price-responses}:
\begin{equation}
 \left(1+\frac{\varepsilon_L\theta}{\sigma}\right)\wh L
 =
 \varepsilon_L\wh y
 +\frac{\varepsilon_L\theta}{\sigma}\wh K.
 \label{eq:A-labor-change}
\end{equation}
The reform lowers \(\tau_L\), so \(\wh y>0\). Both terms on the right are nonnegative and the first is strictly positive; hence labor rises. Since both factor shares are positive, \(\wh Y=\theta\wh K+(1-\theta)\wh L>0\). For current surplus, use the household first-order conditions to write
\begin{align}
 \dd S
 &=(F_K-\Psi_K')\dd K+(F_L-\Psi_L')\dd L \notag\\
 &=r[1-m(q)x]\dd K+w\tau_L\dd L.
 \label{eq:A-surplus}
\end{align}
The capital term is nonnegative, and the labor term is strictly positive because \(\tau_L>0\) and \(\dd L>0\). This proves the surplus claim.
\end{proof}

\section{Proofs for the allocation results}
\label{app:allocation}

\begin{proof}[Proof of Lemma \ref{lem:allocation-tax-equivalence}]
Under full compliance, Euler's theorem and household optimality allow statutory revenue to be expressed entirely in terms of the allocation:
\begin{align*}
 \tau_KF_KK+\tau_LF_LL
 &=
 F-(1-\tau_K)F_KK-(1-\tau_L)F_LL\\
 &=
 F-K\Psi_K'(K)-L\Psi_L'(L).
\end{align*}
Setting this expression equal to \(G\) proves necessity of \eqref{eq:fc-implementability}.

Conversely, take an allocation on the selected branch and set taxes according to \eqref{eq:implementing-taxes}. Conditional on compliance, the two household first-order conditions are satisfied. Strict convexity makes each household's choice optimal and unique at the implied competitive prices. The same revenue identity shows that these taxes raise exactly \(G\). At an interior allocation, the household conditions determine both net-of-tax rates uniquely. Finally, the executive is willing to comply precisely when \eqref{eq:allocation-ic} holds.
\end{proof}

\begin{proof}[Proof of Proposition \ref{prop:allocation-ceiling}]
Using the capital-supply condition at \(q=1\), the amount available for seizure is
\[
 \xi(1-\tau_K)F_KK
 =
 \xi K\Psi_K'(K)
 =
 \xi\phi_K(K).
\]
The executive's incentive constraint is therefore \eqref{eq:allocation-ic}. Since \(\phi_K\) is continuous and strictly increasing, even the smallest capital allocation on the branch is infeasible if \(b<\phi_K(\underline K_G)\). If this inequality does not hold, the feasible capital choices form the interval
\begin{equation}
 [\underline K_G,K_G^D].
 \label{eq:A-feasible-interval}
\end{equation}
If \(K_G^R\) belongs to this interval, it is the unique surplus maximizer. If it lies above the upper endpoint, surplus is strictly increasing throughout the feasible interval and is uniquely maximized at \(K_G^D\). This proves the allocation rule and the stated condition for the ordinary Ramsey allocation to be credible.
\end{proof}

\begin{proof}[Proof of Corollary \ref{cor:allocation-comparative-statics}]
When the ceiling binds,
\[
 \phi_K(K_G^C)=\frac{\beta_O\Delta(p)}{\xi}.
\]
Implicit differentiation, using \(\phi_K'>0\), gives \eqref{eq:ceiling-positive}--\eqref{eq:ceiling-exposure}. A relaxation of the binding ceiling moves capital toward \(K_G^R\). By Assumption \ref{ass:fc-branch}, surplus is strictly increasing on that part of the branch, so a strict increase in feasible capital strictly raises surplus until the Ramsey point is reached. This argument uses the ordering of surplus, not a claim that its derivative must be strictly positive everywhere.

Under strict slack the Ramsey allocation remains feasible after sufficiently small parameter changes and hence does not move. At the boundaries, the same comparisons apply from the feasible side, giving the one-sided statements.
\end{proof}

\section{Proofs for the dynamic results}
\label{app:dynamics}

\begin{proof}[Proof of Proposition \ref{prop:information}]
Keep the prior and continuation schedules fixed. Along the mixing constraint, a rise in compliance lowers the posterior and gives \(\dd\log X/\dd\log q=-\eta\). Differentiating expected diversion therefore yields
\[
 D_q=-X+(1-q)X_q
 =-X\left[1+\eta\frac{1-q}{q}\right]<0.
\]
For the continuation term, differentiate both the probability of compliance and the posterior conditional on it:
\[
 \cJ_q=\cV(z)-z\cV'(z)-\cV(0).
\]
If \(\cV\) is convex, its graph lies above the tangent at \(z\), so \(\cV(0)\geq\cV(z)-z\cV'(z)\). This proves \(\cJ_q\leq0\). The tangent inequality reverses under concavity, and strict curvature makes the respective inequality strict.
\end{proof}

\begin{proof}[Proof of Proposition \ref{prop:continuation-weight}]
Take two policies, \(\chi_2>\chi_1\), and two continuation weights, \(\omega_2>\omega_1\). The change in the relative payoff of the capital-heavy policy is
\begin{align}
 &[W_G(\chi_2;\omega_2)-W_G(\chi_1;\omega_2)]
 -[W_G(\chi_2;\omega_1)-W_G(\chi_1;\omega_1)]
 \notag\\
 &\qquad
 =
 (\omega_2-\omega_1)[j_G(\chi_2)-j_G(\chi_1)]
 \leq0.
 \label{eq:A-decreasing-differences}
\end{align}
Continuity and compactness ensure that the maximizer set is nonempty and has smallest and largest elements. Suppose an optimal policy at the higher weight were strictly above an optimal policy at the lower weight. The two optimality inequalities, together with \eqref{eq:A-decreasing-differences}, imply that both policies must be optimal at both weights. Such a crossing is impossible for either the smallest or the largest maximizer. This proves their weak monotonicity.

If \(j_G\) is strictly decreasing between the two policies, the right-hand side of \eqref{eq:A-decreasing-differences} is strictly negative. The two optimality inequalities would instead require it to be nonnegative. No crossing by any optimal selections is then possible.
\end{proof}

\begin{proof}[Proof of Proposition \ref{prop:common-beta}]
First differentiate the mixing condition with respect to the policy, holding the common discount factor fixed:
\[
 -g+a\frac{q_\chi}{q}
 =
 -\eta\frac{q_\chi}{q},
\]
Rearranging gives \(q_\chi=qg/H\). Next hold the policy fixed and differentiate with respect to \(b\):
\[
 a\frac{q_b}{q}
 =
 1-\eta\frac{q_b}{q},
\]
Thus \(q_b=q/H\). These are different responses, although both use the same local incentive denominator.

The authority's marginal return to a policy change is
\[
 \Pi_\chi=C_\chi+(C_q+\beta J_q)q_\chi.
\]
Differentiate this expression with respect to \(b\), remembering both the induced change in \(q\) and \(\partial\beta/\partial b=\beta\). Collecting the terms gives \eqref{eq:common-beta-cross}. At the stated interior optimum, the strict second-order condition allows the first-order condition to be differentiated locally:
\[
 \frac{\dd\chi^*}{\dd b}
 =-\frac{\Pi_{\chi b}}{\Pi_{\chi\chi}},
\]
The denominator is negative by assumption, so the derivative has the sign of \(\Pi_{\chi b}\), as asserted. Conditions \eqref{eq:common-sufficient-one}--\eqref{eq:common-sufficient-three}, together with \(q_\chi,q_b>0\), make each term in the cross derivative nonpositive.

For the global statement, use the finite-difference inequality \eqref{eq:common-global-dd} on the common compact policy set. The optimality comparison in the proof of Proposition \ref{prop:continuation-weight} then applies directly. Checking cross derivatives only inside a smooth mixing region would not suffice if the comparison also crossed a regime boundary.
\end{proof}

\section{Derivations for the terminal benchmark}
\label{app:terminal}

\begin{proof}[Derivation of \eqref{eq:terminal-gamma}--\eqref{eq:terminal-delta}]
With isoelastic supplies, the terminal household conditions can be written as
\begin{equation}
 K_2
 =
 \bar K_2
 \left(z\bar x\alpha\frac{Y_2}{K_2}\right)^{\bar\varepsilon_K},
 \qquad
 L_2
 =
 \bar L_2
 \left(\bar y(1-\alpha)\frac{Y_2}{L_2}\right)^{\bar\varepsilon_L},
 \label{eq:A-terminal-focs}
\end{equation}
where \(\bar K_2,\bar L_2>0\) collect the supply scale parameters. Solving for each input as a function of output and substituting into Cobb--Douglas production gives
\[
 Y_2=C(z\bar x)^\gamma\bar y^\delta,
\]
where \(\gamma\) is given in \eqref{eq:terminal-gamma} and
\[
 \delta=
 \frac{(1-\alpha)\bar\varepsilon_L(1+\bar\varepsilon_K)}
 {1+(1-\alpha)\bar\varepsilon_K+\alpha\bar\varepsilon_L}.
\]
The terminal net-of-tax rates and all scale parameters are fixed, so they can be collected in \(\bar Y\), leaving \(Y_2=\bar Yz^\gamma\). The opportunist receives \(\bar x\alpha Y_2\) upon seizure. At zero reputation there is no terminal output, so subtracting \(V_O(0)=0\) gives \eqref{eq:terminal-opportunist}--\eqref{eq:terminal-delta}.
\end{proof}

\begin{proof}[Proof of Proposition \ref{prop:primitive-eta}]
Since \(\eta=\gamma\) and the denominator in \eqref{eq:terminal-gamma} is positive, \(\eta\leq1\) is equivalent to \eqref{eq:gamma-condition}. To obtain the sufficient restriction, subtract the left-hand side of that inequality from its right-hand side:
\[
 1+(1-2\alpha)\bar\varepsilon_K
 +\alpha\bar\varepsilon_L(1-\bar\varepsilon_K),
\]
Both nonconstant terms are nonnegative when \(\alpha\leq1/2\) and \(\bar\varepsilon_K\leq1\), whatever the positive labor-supply elasticity. This proves the sufficient condition.
\end{proof}

\begin{proof}[Proof of Proposition \ref{prop:terminal-welfare}]
Integrating the isoelastic supply schedules, with costs normalized to zero at zero supply, and applying \eqref{eq:A-terminal-focs} gives
\[
 \Psi_{K,2}(K_2)
 =
 \frac{\bar\varepsilon_K}{1+\bar\varepsilon_K}
 z\bar x\alpha Y_2,
 \qquad
 \Psi_{L,2}(L_2)
 =
 \frac{\bar\varepsilon_L}{1+\bar\varepsilon_L}
 \bar y(1-\alpha)Y_2.
\]
Subtracting these costs from output and using \(Y_2=\bar Yz^\gamma\) proves \eqref{eq:terminal-surplus}--\eqref{eq:terminal-cost-shares}. The constants satisfy \(0<\ell<1\), \(k>0\), and \(k+\ell<1\). Differentiating twice yields
\begin{equation}
 S_2''(z)
 =
 \bar Y\left[
 (1-\ell)\gamma(\gamma-1)z^{\gamma-2}
 -k\gamma(\gamma+1)z^{\gamma-1}
 \right]<0
 \label{eq:A-terminal-second}
\end{equation}
for \(0<\gamma\leq1\): the first term is nonpositive and the second is strictly negative. Since \(S_2(0)=0\), the strict concavity inequality gives \(S_2(z)-zS_2'(z)>0\). This is exactly the derivative of \(qS_2(p/q)\) with respect to \(q\). The current-surplus comparison follows from Corollary \ref{cor:current-surplus} under its stated conditions.
\end{proof}


\begin{thebibliography}{99}

\bibitem[Ablyatifov and Lukyanov(2026)]{AblyatifovLukyanov2026}
Ablyatifov, E. and Lukyanov, G. (2026).
Government reputation and fiscal capacity.
Working paper, July 2026. arXiv:2509.03087.

\bibitem[Acemoglu et~al.(2011)Acemoglu, Golosov, and Tsyvinski]{AcemogluGolosovTsyvinski2011}
Acemoglu, D., Golosov, M., and Tsyvinski, A. (2011).
Political economy of Ramsey taxation.
\textit{Journal of Public Economics} 95(7--8), 467--475.

\bibitem[Aguiar and Amador(2011)]{AguiarAmador2011}
Aguiar, M. and Amador, M. (2011).
Growth in the shadow of expropriation.
\textit{Quarterly Journal of Economics} 126(2), 651--697.

\bibitem[Benhabib and Rustichini(1997)]{BenhabibRustichini1997}
Benhabib, J. and Rustichini, A. (1997).
Optimal taxes without commitment.
\textit{Journal of Economic Theory} 77(2), 231--259.

\bibitem[Diamond and Mirrlees(1971a)]{DiamondMirrlees1971a}
Diamond, P. A. and Mirrlees, J. A. (1971a).
Optimal taxation and public production I: Production efficiency.
\textit{American Economic Review} 61(1), 8--27.

\bibitem[Diamond and Mirrlees(1971b)]{DiamondMirrlees1971b}
Diamond, P. A. and Mirrlees, J. A. (1971b).
Optimal taxation and public production II: Tax rules.
\textit{American Economic Review} 61(3), 261--278.

\bibitem[Dovis and Kirpalani(2021)]{DovisKirpalani2021}
Dovis, A. and Kirpalani, R. (2021).
Rules without commitment: Reputation and incentives.
\textit{Review of Economic Studies} 88(6), 2833--2856.

\bibitem[Halac and Yared(2022)]{HalacYared2022}
Halac, M. and Yared, P. (2022).
Fiscal rules and discretion under limited enforcement.
\textit{Econometrica} 90(5), 2093--2127.

\bibitem[Lu(2013)]{Lu2013}
Lu, Y. K. (2013).
Optimal policy with credibility concerns.
\textit{Journal of Economic Theory} 148(5), 2007--2032.

\bibitem[Park(2014)]{Park2014}
Park, Y. (2014).
Optimal taxation in a limited commitment economy.
\textit{Review of Economic Studies} 81(2), 884--918.

\bibitem[Phelan(2006)]{Phelan2006}
Phelan, C. (2006).
Public trust and government betrayal.
\textit{Journal of Economic Theory} 130(1), 27--43.

\bibitem[Phelan and Stacchetti(2001)]{PhelanStacchetti2001}
Phelan, C. and Stacchetti, E. (2001).
Sequential equilibria in a Ramsey tax model.
\textit{Econometrica} 69(6), 1491--1518.

\bibitem[Ramsey(1927)]{Ramsey1927}
Ramsey, F. P. (1927).
A contribution to the theory of taxation.
\textit{Economic Journal} 37(145), 47--61.

\bibitem[Reis(2013)]{Reis2013}
Reis, C. (2013).
Taxation without commitment.
\textit{Economic Theory} 52(2), 565--588.

\bibitem[Scheuer and Wolitzky(2016)]{ScheuerWolitzky2016}
Scheuer, F. and Wolitzky, A. (2016).
Capital taxation under political constraints.
\textit{American Economic Review} 106(8), 2304--2328.

\bibitem[Yun(2026)]{Yun2026}
Yun, Y. (2026).
Government reputation, FDI, and profit-shifting.
Working paper, July 2026 revision. SSRN 4569821.

\end{thebibliography}
\end{document}